\documentclass{article}

\usepackage[final]{nips_2018}

\usepackage[utf8]{inputenc} 
\usepackage[T1]{fontenc}    
\usepackage{hyperref}       
\usepackage{url}            
\usepackage{booktabs}       
\usepackage{amsfonts}       
\usepackage{nicefrac}       
\usepackage{microtype}      
\usepackage{algorithm,algorithmic}
\usepackage{amsmath,amsthm}
\usepackage{bbm}
\usepackage{mathtools}

\newcommand{\mH}{\mathcal{H}}

\newcommand{\trace}{\operatorname{Tr}}
\newtheorem{theorem}{Theorem}
\newtheorem{claim}[theorem]{Claim}
\newtheorem{lemma}[theorem]{Lemma}
\newtheorem{corollary}[theorem]{Corollary}
\newtheorem{definition}{Definition}

\newcommand{\mycases}[4]{{
\left\{
\begin{array}{ll}
    {#1} & {\;\text{#2}} \\[1ex]
    {#3} & {\;\text{#4}}
\end{array}
\right. }}
\newcommand{\D}{\mathcal{D}}
\newcommand{\K}{\ensuremath{\mathcal K}}
\newcommand{\reals}{\mathbb{R}}
\newcommand{\E}{\mathbb{E}}
\DeclareMathOperator*{\argmin}{arg\,min}

\title{Online Learning of Quantum States}

\author{Scott Aaronson\\ 
UT Austin \thanks{Supported
by a Vannevar Bush Faculty Fellowship from the US Department of Defense.
\ Part of this work was done while the author was supported by an NSF Alan T.
Waterman Award.}\\
\texttt{aaronson@cs.utexas.edu}\\
\And Xinyi Chen\\
Princeton University and Google AI Princeton  \\
\texttt{xinyic2015@gmail.com} \\
\And Elad Hazan\\
Princeton University and Google AI Princeton \\
\texttt{ehazan@cs.princeton.edu} \\
\And Satyen Kale \\
Google AI, New York\\
\texttt{satyenkale@google.com}\\
\And Ashwin Nayak \\
University of Waterloo \thanks{Research supported in part by 
NSERC Canada.}\\
\texttt{ashwin.nayak@uwaterloo.ca}
}
\hypersetup{draft}

\begin{document}

\maketitle

\begin{abstract}
Suppose we have many copies of an unknown $n$-qubit state $\rho$. \ We measure
some copies of\ $\rho$\ using a known two-outcome measurement $E_{1}$, then
other copies using a measurement $E_{2}$, and so on. \ At each stage $t$, we
generate a current hypothesis $\omega_{t}$ about the state $\rho$, using the
outcomes of the previous measurements. \ We show that it is possible to do this
in a way that guarantees that $\left\vert  \trace \!\left(  E_{i}%
\omega_{t}\right)  - \trace \!\left(  E_{i}\rho\right)  \right\vert $,
the error in our prediction for the next measurement, is at least $\varepsilon$ at most
$\operatorname{O}\!\left( n / \varepsilon^2 \right)  $\ times. \ Even in the \textquotedblleft
non-realizable\textquotedblright\ setting---where there could be arbitrary noise in
the measurement outcomes---we show how to output hypothesis states that incur at most  $\operatorname{O}(  \sqrt
{Tn}\;)  $ excess loss over the best possible state on the first $T$\ measurements. \ These results
generalize a 2007 theorem by Aaronson on the PAC-learnability of
quantum states, to the online and regret-minimization settings. \ We give three different ways to
prove our results---using convex optimization, quantum postselection, and sequential fat-shattering dimension---which have different advantages in terms of parameters and portability.

\end{abstract}


\section{Introduction}
\label{sec-intro}

\emph{State tomography\/} is a fundamental task in quantum computing
of great practical and theoretical importance. In a typical scenario, we 
have access to an apparatus that is capable of producing many copies of 
a quantum state, and we wish to obtain a description of the state via
suitable measurements. Such
a description would allow us, for example, to check the accuracy with which 
the apparatus constructs a specific target state.

How many single-copy measurements 
are needed to \textquotedblleft learn\textquotedblright%
\ an unknown $n$-qubit quantum state $\rho$? \ Suppose we wish to reconstruct 
the full $2^{n}\times2^{n}$ density matrix, even approximately, to
within~$\varepsilon$ in trace distance. If we make no
assumptions about $\rho$, then it is straightforward to show that the 
number of measurements needed grows exponentially with $n$. In fact, even
when we allow joint measurement of multiple copies of the state, an
exponential number of copies of~$\rho$ are required (see, e.g., 
\cite{OW16,HHJWY17}). (A ``joint measurement'' of two or more states on 
disjoint sequences of qubits is a \emph{single\/} measurement of all the 
qubits together.)

Suppose, on the other hand, that there is some
probability distribution $\mathcal{D}$ over possible yes/no measurements,
where we identify the measurements with $2^{n}\times2^{n}$\ Hermitian matrices $E$\ with
eigenvalues in $\left[  0,1\right]  $. \ Further suppose we are only concerned about
learning the state\ $\rho$\ well enough to predict the outcomes of \textit{most}%
\ measurements~$E$\ drawn from $\mathcal{D}$---where \textquotedblleft
predict\textquotedblright\ means approximately calculating the probability,
$\trace \!\left(  E\rho\right)  $, of a \textquotedblleft
yes\textquotedblright\ result. \ Then for how many (known) sample measurements
$E_{i}$, drawn independently from~$\mathcal{D}$, do we need to know the
approximate value of $\trace \!\left(  E_{i}\rho\right)  $, before we
have enough data to achieve this?

\cite{aar:learn}\ proved that the number of sample
measurements needed, $m$, grows only \textit{linearly} with the number of
qubits $n$. \ What makes this surprising is that it represents an exponential
reduction compared to full quantum state tomography. \ Furthermore, the
prediction strategy is extremely simple. \ Informally, we merely need to find
any \textquotedblleft hypothesis state\textquotedblright\ $\omega$\ that
satisfies $\trace \!\left(  E_{i}\omega\right)  \approx
\trace \!\left(  E_{i}\rho\right)  $\ for all the sample measurements
$E_{1},\ldots,E_{m}$.\ \ Then with high probability over the choice of sample
measurements, that hypothesis $\omega$\ necessarily \textquotedblleft
generalizes\textquotedblright,\ in the sense that $\trace \!\left(
E\omega\right)  \approx\trace \!\left(  E\rho\right)  $\ for most
additional $E$'s drawn from $\mathcal{D}$. \ The learning theorem led to
followup work including a full characterization of
quantum advice (\cite{adrucker}); efficient learning for stabilizer states (\cite{rocchetto});
the ``shadow tomography'' protocol (\cite{aar:shadow}); and recently, the first experimental demonstration of quantum state PAC-learning (\cite{explearning}).

A major drawback of the learning theorem due to Aaronson is the assumption that the
sample measurements are drawn \textit{independently} from $\mathcal{D}$---and
moreover, that the same distribution $\mathcal{D}$\ governs both the training
samples, and the measurements on which the learner's performance is later
tested. \ It has long been understood, in computational learning theory, that
these assumptions are often unrealistic: they fail to account for adversarial
environments, or environments that change over time. This is precisely
the state of affairs in current experimental implementations of quantum
information processing. Not all measurements of quantum states may be
available or feasible in a specific implementation, \emph{which\/} measurements 
are feasible is dictated by Nature, and as we develop more control over the
experimental set-up, more sophisticated measurements become available.
The task of learning a state prepared in the laboratory thus takes the
form of a game, with the theorist on one side, and the experimentalist 
and Nature on the other: the theorist is repeatedly challenged to predict 
the behaviour of the state with respect to the next measurement that 
Nature allows the experimentalist to realize, with the opportunity to 
refine the hypothesis as more measurement data become available.

It is thus desirable to design learning algorithms that work in the
more stringent \textit{online learning} model. \ Here the learner is
presented a sequence of input points, say $x_{1},x_{2},\ldots$, one at a time.
\ Crucially, there is no assumption whatsoever about the $x_{t} $'s: the
sequence could be chosen adversarially, and even adaptively, which means that
the choice of~$x_{t}$\ might depend on the learner's behavior on $x_{1}%
,\ldots,x_{t-1}$. \ The learner is trying to learn some unknown function
$f\!\left(  x\right)  $, about which it initially knows only that $f$ belongs to
some hypothesis class $\mH$---or perhaps not even that; we also consider the scenario where
the learner simply tries to compete with the best predictor in $\mH$, which might or might not be a good predictor. \ The learning proceeds as follows: for
each $t$, the learner first guesses a value $y_{t}$\ for $f\!\left(
x_{t}\right)  $, and is then told the true value~$f\!\left(  x_{t}\right)  $,
or perhaps only an approximation of this value. 
\ Our goal is to design a learning algorithm with the following guarantee:
\textit{regardless of the sequence of }$x_{t}$\textit{'s, the learner's guess,
}$y_{t}$\textit{, will be far from the true value }$f\!\left(  x_{t}\right)
$\textit{\ at most }$k$\textit{\ times} (where $k$, of course, is
as small as possible). \ The $x_{t}$'s on which the learner errs could be
spaced arbitrarily; all we require is that they be bounded in number.

This leads to the following question: can the learning theorem established
by
\cite{aar:learn}\ be generalized to the online learning setting? \ In other
words: is it true that, given a sequence $E_{1},E_{2},\ldots$\ of yes/no
measurements, where each $E_{t}$\ is followed shortly afterward by an
approximation of
$\trace \!\left(  E_{t}\rho\right)  $, there is a way to anticipate
the $\trace \!\left(  E_{t}\rho\right)  $\ values by guesses $y_{t}%
\in\left[  0,1\right]  $, in such a way that $\left\vert y_{t}%
-\trace \!\left(  E_{t}\rho\right)  \right\vert >\varepsilon$\ at most, say,
$\operatorname{O}\!\left(  n\right)  $\ times (where $\varepsilon>0$\ is some constant, and
$n$\ again is the number of qubits)?
The purpose of this paper is to provide an affirmative answer.

Throughout the paper, we consider only two-outcome measurements of an
$n$ qubit mixed state $\rho$, and we specify such a measurement by a
$2^{n}\times2^{n}$ Hermitian matrix $E$ with
eigenvalues in $\left[  0,1\right]$.
We say that~$E$
\textquotedblleft accepts\textquotedblright\ $\rho$\ with probability
$\trace \!\left(  E\rho\right)  $\ and \textquotedblleft
rejects\textquotedblright\ $\rho$ with probability $1-\trace \!\left(
E\rho\right)  $. We prove that:

\begin{theorem}
\label{thm1}
Let $\rho$\ be an $n$-qubit mixed state, and let $E_{1},E_{2},\ldots$\ be a
sequence of $2$-outcome measurements that are revealed to the learner one by one, each followed by a value $b_t \in [0,1]$ such that $\left| \trace(E_t \rho) - b_t \right| \leq \varepsilon / 3$.  \ Then there is an explicit strategy for outputting hypothesis states
$\omega_{1},\omega_{2},\ldots$\ such that
$
\left\vert \trace \!\left(  E_t \omega_{t}\right)
-\trace \!\left(  E_t \rho\right)  \right\vert >\varepsilon
$
for at most $\operatorname{O}\!\left( \frac{ n}{\varepsilon^2} \right)  $\ values of $t$.
\end{theorem}

We also prove a theorem for the so-called \emph{regret minimization model\/} (i.e., the ``non-realizable case''), where we make no assumption about the input data arising from an actual quantum state, and our goal is simply to do not much worse than the best hypothesis state that could be found with perfect foresight.
In this model, the measurements~$E_1, E_2, \dotsc$ are presented to a
learner one-by-one. In iteration~$t$, after seeing $E_t$, the learner is 
challenged to output a hypothesis state $\omega_t$, and then suffers a ``loss'' 
equal to $\ell_t(\trace(E_t \omega_t))$ where $\ell_t$ is a real
function that is revealed to the learner. Important examples of loss 
functions are~$L_1$ loss, when~$\ell_t(z) \coloneqq \left| z - b_t
\right|$, and~$L_2$ loss, when~$\ell_t(z) \coloneqq \left( z - b_t 
\right)^2$, where~$b_t \in [0,1]$. The number~$b_t$ may be an
approximation of~$\trace(E_t \rho)$ for some fixed but unknown quantum 
state~$\rho$, but is allowed to be arbitrary in general. In particular, the
pairs~$(E_t, b_t)$ may not be consistent with any quantum state.
Define the \emph{regret\/}~$R_T$, after $T$ iterations, to be the amount
by which the actual loss of the learner exceeds the loss of the best 
single hypothesis:
$$
R_T \coloneqq \sum_{t=1}^{T} \ell_t(\trace(E_t \omega_t))
- \min_{\varphi} ~ \sum_{t=1}^{T} \ell_t(\trace(E_t \varphi)) \enspace.
$$
The learner's objective is to minimize regret. We show that:
\begin{theorem}
\label{thm2} Let $E_{1},E_{2},\ldots$\ be a sequence of two-outcome 
measurements on an $n$-qubit state presented to the learner,
and~$\ell_1, \ell_2, \dotsc$ be the corresponding loss functions
revealed in successive iterations in the regret minimization model.
Suppose~$\ell_t$ is convex and $L$-Lipschitz; in particular, for 
every $x \in \reals$, there is a sub-derivative $\ell_t'(x)$ such
that $\left| \ell_t'(x) \right| \leq L$.
Then there is an explicit learning strategy that guarantees
regret~$R_T = \operatorname{O} ( L\sqrt{Tn}\;)$ for
all $T$. \ This is so even assuming the measurement~$E_t$ and loss
function~$\ell_t$ are chosen adaptively, in response to the learner's 
previous behavior.

Specifically, the algorithm applies to $L_1$ loss and $L_2$ loss, and achieves
regret~$\operatorname{O}( \sqrt{Tn}\;)$ for both.
\end{theorem}

The online strategies we present enjoy several advantages over full
state tomography, and even over ``state certification'', in which we
wish to test whether a given quantum state is close to a desired state or far
from it. Optimal algorithms for state tomography~(\cite{OW16,HHJWY17})
or certification~(\cite{BOW17})
require joint measurements of an exponential number of copies of the
quantum state, and assume the ability to perform noiseless,
universal quantum computation. On the other hand, the algorithms
implicit in Theorems~\ref{thm1} and~\ref{thm2} involve only single-copy
measurements, allow for noisy measurements, and capture ground reality more
closely. They produce a hypothesis state that mimics the unknown state 
with respect to measurements
that \emph{can be\/} performed in a given experimental set-up, and the accuracy
of prediction improves as the set of available measurements grows.
For example, in the realizable case, i.e., when the data arise from an actual
quantum state, the average $L_1$ loss per iteration is $\operatorname{O} 
(\sqrt{n/T}\,)$. This tends to zero, as the number of
measurements becomes large. Note that $L_1$ loss may be as large
as~$1/2$ per iteration in the worst case, but this occurs at
most~$\operatorname{O} (\sqrt{nT}\, )$ times.
Finally, the algorithms have run time exponential in the number of qubits 
in each iteration, but are entirely \emph{classical\/}.
Exponential run time is unavoidable, as the measurements are presented 
explicitly as~$2^n \times 2^n$ matrices, where~$n$ is the number of 
qubits. If we were required to output the hypothesis states, the length 
of the output---also exponential in the number of qubits---would again
entail exponential run time.

It is natural to wonder whether Theorems \ref{thm1} and \ref{thm2} leave any room for improvement. \ Theorem \ref{thm1} is asymptotically optimal in its mistake bound of $\operatorname{O}(n/\varepsilon^2)$;
this follows from the property that $n$-qubit quantum states, considered as a hypothesis class, have $\varepsilon$-fat-shattering dimension $\Theta(n/\varepsilon^2)$ (see, for example, \cite{aar:learn}).
On the other hand, there is room to improve Theorem \ref{thm2}. The 
bounds of which we are aware are $\Omega(\sqrt{T n}\,)$ for the $L_1$ loss~(see,
e.g., \cite[Theorem~4.1]{mmw}) in the non-realizable case and $\Omega(n)$ for the $L_2$ loss
in the realizable case,
when the feedback consists of the measurement outcomes. \ (The latter
bound, as well as an~$\Omega(\sqrt{T n}\,)$ bound for~$L_1$ loss in the same
setting, come from considering quantum mixed states that consist of $n$ independent classical coins, each of which could land heads with probability either $1/2$ or $1/2 + \varepsilon$.
The paramater~$\varepsilon$ is set to~$\sqrt{n/T}$.)

We mention an application of Theorem \ref{thm1}, that appears in simultaneous work. \ \cite{aar:shadow} has given an algorithm for the so-called \emph{shadow tomography} problem. \ Here we have an unknown $D$-dimensional pure state $\rho$, as well as known two-outcome measurements $E_1,\ldots,E_m$. \ Our goal is to approximate $\trace(E_i \rho)$, for \emph{every} $i$, to within additive error $\varepsilon$. \ We would
like to do this by measuring $\rho^{\otimes k}$, where $k$ is as small as possible. \ Surprisingly, \cite{aar:shadow} showed that this can be achieved with $k = \widetilde{\operatorname{O}}( (\log M)^4 (\log D) / \varepsilon^5 )$, that is, a number of copies of $\rho$ that
is only \emph{polylogarithmic} in both $D$ and $M$. \ One component of his algorithm is essentially tantamount to online learning with $\widetilde{\operatorname{O}}(n / \varepsilon^3)$
mistakes---i.e., the learning algorithm we present in Section \ref{sec:postselection} of this paper. \ However, by using Theorem \ref{thm1} from this paper in a black-box manner, we can improve the sample complexity of shadow tomography to $\widetilde{\operatorname{O}}( (\log M)^4 (\log D) / \varepsilon^4 )$. \ Details appear in (\cite{aar:shadow}).

To maximize insight, in this paper we give \textit{three} very different approaches to proving Theorems \ref{thm1} and \ref{thm2} (although we do
not prove every statement with all three approaches).
Our first approach is to adapt techniques from online convex optimization to the setting of density matrices,
which in general may be over a complex Hilbert space. \ This requires extending standard techniques to cope with convexity and Taylor approximations, which are widely used for functions over the real domain, but not over the complex domain. \ We also give an efficient iterative algorithm to produce predictions. \ This approach connects our problem to the modern mainstream of online learning algorithms, and achieves the best parameters
(as stated in Theorems~\ref{thm1} and~\ref{thm2}).

Our second approach is via a postselection-based learning procedure, which starts with the maximally mixed state as a hypothesis and then repeatedly refines it by simulating postselected measurements. \ This approach builds on earlier work due
to \cite{aar:adv}, specifically the proof of $\mathsf{BQP/qpoly} \subseteq \mathsf{PP/poly}$. \ The advantage is that it
is almost entirely self-contained, requiring no ``power tools'' from convex optimization or learning theory. \ On the other hand, the approach does
not give optimal parameters, and we do not know how to prove Theorem \ref{thm2} with it.

Our third approach is via an upper-bound on the so-called \textit{sequential fat-shattering dimension} of quantum states, considered as a hypothesis class
(see, e.g., \cite{rst}). \ In the original quantum PAC-learning theorem
by Aaronson, the key step was to upper-bound the so-called $\varepsilon$-\textit{fat-shattering dimension} of quantum states considered as a hypothesis class. \ Fat-shattering dimension is a real-valued generalization of VC dimension. \ One can then appeal to known results to get a sample-efficient learning algorithm. \ For online learning, however, bounding the fat-shattering dimension no longer suffices; one instead needs to consider a possibly-larger quantity called sequential fat-shattering dimension. \ However, by appealing
to a lower bound due
to~\cite{Nayak99,antv} for a variant of quantum \emph{random access codes\/}, we
are able to upper-bound the sequential fat-shattering dimension of quantum states. \ Using known results---in particular, those due
to \cite{rst}---this implies the regret bound in Theorem \ref{thm2}, up to a multiplicative
factor of~$\log^{3/2} T$. \ The statement that \textit{the hypothesis class of $n$-qubit states has $\varepsilon$-sequential fat-shattering dimension $\operatorname{O}
(n/\varepsilon^2)$} might be of independent interest: among other things, it implies that \emph{any\/} online learning algorithm that works given bounded sequential fat-shattering dimension, will work for online learning of quantum states.
We also give an alternative proof for the lower bound due to Nayak for 
quantum random access codes, and extend it to codes that are decoded by
what we call \emph{measurement decision trees\/}. We expect these also
to be of independent interest.

\subsection{Structure of the paper}

We start by describing background and the technical learning setting as well as notations used throughout
(Section~\ref{sec:definitions}). In Section \ref{sec:online} we give the algorithms and main theorems derived using convexity arguments and online convex optimization. \ In Section \ref{sec:postselection} we state
the main theorem using a postselection algorithm. \ In Section \ref{sec:dimension} we give a sequential fat-shattering dimension bound for quantum states and its implication for online learning of quantum states.
Proofs of the theorems and related claims are presented in the
appendices.

\section{Preliminaries and definitions}  \label{sec:definitions}

We define the trace norm of a matrix~$M$ as~$\left\Vert M \right\Vert
_{\trace} \coloneqq \trace \sqrt{M M^\dagger}$,
where~$M^\dagger$ is the adjoint of~$M$.
We denote the~$i$th eigenvalue of a Hermitian matrix~$X$
by~$\lambda_i(X)$, its minimum eigenvalue by~$\lambda_{\min}(X)$, and
its maximum eigenvalue by~$\lambda_{\max}(X)$.  We sometimes use the 
notation~$X \bullet Y$ to denote the trace 
inner-product~$\trace(X^\dagger Y)$ between two complex matrices of the
same dimensions. By `$\log$' we denote the natural logarithm, unless 
the base is explicitly mentioned.

An~$n$-qubit quantum state~$\rho$ is an element of~$C_n$, where $C_n$ is 
the set of all trace-1 positive
semi-definite (PSD) complex matrices of dimension $2^n$:
$$ C_n = \{ M \in  \mathbb{C}^{2^n \times 2^n} \ , \  M = M^{\dagger}  \ , \ M \succeq 0 \ , \ \trace(M)  = 1 \} \enspace. $$
Note that $C_n$ is a convex set.\ A two-outcome measurement of an~$n$-qubit state is defined by a $2^n \times 2^n$ Hermitian matrix $E$ with eigenvalues in $\lbrack 0, 1\rbrack$.
The measurement $E$ ``accepts'' $\rho$ with probability $\trace(E\rho)$,
and ``rejects'' with probability $1-\trace(E\rho)$. 
For the algorithms we present in this article, we assume that a two-outcome 
measurement is specified via a classical description of its defining
matrix~$E$.
In the rest of the article, unless mentioned otherwise, a ``measurement'' 
refers to a ``two-outcome measurement''.
We refer the reader to the book by \cite{Watrous18} for a more thorough
introduction to the relevant concepts from quantum information.

\paragraph{Online learning and regret.}  In online learning of quantum states, 
we have a sequence of iterations~$t = 1, 2, 3, \dotsc$ of the following form.
First, the learner constructs a state $\omega_t \in C_n$; we say that the 
learner ``predicts''~$\omega_t$. It then suffers a
``loss''~$\ell_t(\trace(E_t\omega_t))$ that depends on a measurement $E_t$,
both of which are
presented by an adversary. \ Commonly used loss functions are $L_2$
loss (also called ``mean square error''), given by
$$
\ell_t(z) \coloneqq (z - b_t)^2 \enspace,
$$
and $L_1$ loss (also called ``absolute loss''), given by
\[
\ell_t(z) \coloneqq \left| z - b_t \right| \enspace,
\]
where~$b_t \in [0,1]$. The parameter~$b_t$ may be an approximation 
of~$\trace(E_t \rho)$ for some fixed quantum state~$\rho$ not known to the 
learner, obtained by measuring multiple copies of~$\rho$.
However, in general, the parameter is allowed to be arbitrary.

The learner then ``observes'' feedback from the measurement~$E_t$;
the feedback is also provided by the adversary.
The simplest feedback is the realization of a binary random variable $Y_t$ such that
$$ Y_t = \mycases {1}{with probability $\trace(E_t \rho)$ \enspace,
\qquad and}{0}{with probability $1 - \trace(E_t \rho )$ \enspace.}
$$
Another common feedback is a number~$b_t$ as described above, especially 
in case that the learner suffers $L_1$ or $L_2$ loss.

We would like to design a strategy for updating $\omega_t$ based on the
loss, measurements, and feedback in all the
iterations so far, so that the learner's total loss is minimized in the 
following sense. We would like that over~$T$ iterations (for a
number~$T$ known in advance), the learner's total loss is not much more 
than that of the hypothetical strategy of outputting
the same quantum state $\varphi$ at every time step, where~$\varphi$
minimizes the total loss \emph{with perfect hindsight\/}.
Formally this is captured by the notion of \textit{regret\/}~$R_T$, defined as
$$ R_T \coloneqq \sum_{t=1}^T \ell_t(\trace(E_t \omega_t))
    - \min_{\varphi \in C_n} \sum_{t=1}^T \ell_t(\trace(E_t \varphi))
    \enspace.
$$
The sequence of measurements~$E_t$ can be arbitrary, even adversarial, based on the learner's previous actions.
Note that if the loss function is given by a fixed state~$\rho$ (as in the case
of mean square error), the minimum total loss would be~$0$. This is
called the ``realizable'' case.  However, in general, the loss function
presented by the adversary need not be consistent with any quantum state.
This is called the ``non-realizable'' case.

A special case of the online learning setting is called \textit{agnostic learning}; here the measurements~$E_t$ are drawn from a fixed and unknown distribution $\D$. \ The setting is called ``agnostic'' because we still do
not assume that the losses correspond to any actual state $\rho$ (i.e.,
the setting may be non-realizable).

\paragraph{Online mistake bounds.}
In some online learning scenarios the quantity of interest is not the
mean square error, or some other convex loss, but rather simply the
total number of ``mistakes'' made. \ For example, we may be interested in
the number of iterations in which the predicted probability of
acceptance~$\trace(E_t \omega_t)$ is more than $\varepsilon$-far from the
actual value~$\trace(E_t \rho)$,
where~$\rho$ is again a fixed state not known to the learner. \ More formally, let
$$
\ell_t(\trace(E_t\omega_t)) \coloneqq |\trace(E_t\omega_t)-\trace(E_t\rho)|
$$
be the absolute loss function.
Then the goal is to bound the number of iterations in which
$\ell_t(\trace(E_t\omega_t)) > \varepsilon$, regardless of the sequence of
measurements~$E_t$ presented by the adversary. \ We assume that in this setting,the adversary provides as feedback an approximation~$b_t \in [0,1]$ that satisfies $|\trace(E_t\rho)-b_t| \leq \frac{\varepsilon}{3}$.

\section{Online learning of quantum states}  \label{sec:online}

In this section, we use techniques from online convex optimization to
minimize regret. The same algorithms may be adapted to also minimize the
number of mistakes made.

\subsection{Regularized Follow-the-Leader}

We first follow the template of the Regularized Follow-the-Leader 
algorithm~(RFTL; see, for example, \cite[Chapter~5]{oco}).
The algorithm below makes use of von Neumann entropy, which relates to the Matrix Exponentiated Gradient algorithm (\cite{tsuda}).

\begin{algorithm}[H] 
\caption{RFTL for Quantum Tomography}
\begin{algorithmic}[1] \label{alg:rftl}
\STATE Input: $T$, $\K \coloneqq C_n$, $\eta < \frac{1}{2}$
\STATE Set $\omega_1  \coloneqq 2^{-n} \mathbbm{I}$.
\FOR{$t = 1, \ldots, T$}
\STATE Predict $\omega_t$. Consider the convex and $L$-Lipschitz loss
function $\ell_t:\reals \to \reals$ given by measurement $E_t:
\ell_t(\trace(E_t \varphi)).$  Let $\ell_t'(x)$ be a sub-derivative of $\ell_t$ with respect to $x$. Define
$$
\nabla_t\coloneqq \ell_t'(\mathrm{Tr}(E_t \omega_t))E_t \enspace.
$$

\STATE Update decision according to the RFTL rule with von Neumann entropy:
\begin{equation} \label{eqn:regret1}
 \omega_{t+1} \coloneqq \argmin _{\varphi \in \K}  \left\{  \eta\sum_{s = 1}^t \trace(\nabla_s \varphi) +  \sum_{i=1}^{2^n} \lambda_i (\varphi) \log \lambda_i(\varphi) \right\}
\enspace.
\end{equation}
\ENDFOR
\end{algorithmic}
\end{algorithm}

\paragraph{Remark 1:}  The mathematical program in Eq.~\eqref{eqn:regret1} is convex, and thus can be solved in polynomial time in the dimension, which is $2^n$.

\begin{theorem} \label{thm:main}
Setting $\eta = \sqrt{\frac{(\log 2)n}{2TL^2}}$~, the regret of
Algorithm~\ref{alg:rftl} is bounded by $2L \sqrt{(2 \log 2) T n}$~.
\end{theorem}

\paragraph{Remark 2:} In the case where the feedback is an independent
random variable~$Y_t$, where~$Y_t = 0$ with probability $1 -
\trace(E_t \rho)$ and $Y_t = 1$ with probability $\trace(E_t \rho)$ for
a fixed but unknown state~$ \rho$, we define~$\nabla_t$ in 
Algorithm~\ref{alg:rftl} as~$\nabla_t \coloneqq 
2 (\trace(E_t \omega_t) - Y_t) E_t$.
Then~$\E[\nabla_t]$ is the gradient of the~$L_2$ loss function where we
receive precise feedback $\trace(E_t \rho)$ instead of~$Y_t$. It follows
from the proof of Theorem~\ref{thm:main} that the expected $L_2$ regret
of Algorithm \ref{alg:rftl}, namely
\[
\E \left[\sum_{t=1}^T (\trace(E_t \omega_t) - \trace(E_t \rho))^2 \right]
\enspace,
\]
is bounded by $\operatorname{O}(\sqrt{Tn}\,)$. 

The proof of Theorem~\ref{thm:main} appears in Appendix B. The proof is along the lines of \cite[Theorem~5.2]{oco}, except that the loss function does not take a raw state as input, and our domain for optimization is complex. Therefore, the mean value theorem does not hold, which means we need to approximate the Bregman divergence instead of replacing it by a norm as in the original proof.  Another subtlety is that convexity needs to be carefully defined with respect to the complex domain.

\subsection{Matrix Multiplicative Weights}
\label{sec-mmw}

The Matrix Multiplicative Weights (MMW) algorithm~\citep{tsuda} provides an alternative means of proving Theorem~\ref{thm2}. The algorithm follows the template of Algorithm~\ref{alg:rftl} with step 5 replaced by the following update rule:
\begin{equation} \label{eq:mmw-update}
	 \omega_{t+1} \coloneqq \frac{\exp(-\tfrac{\eta}{L}\textstyle\sum_{\tau=1}^t \nabla_\tau )}{\trace(\exp(-\tfrac{\eta}{L}\textstyle\sum_{\tau=1}^t \nabla_\tau ))}
\enspace.
\end{equation}
In the notation of~\cite{arora-kale}, this algorithm is derived using
the loss matrices $M_t = \frac{1}{L}\nabla_t =
\frac{1}{L}\ell_t'(\trace(E_t \omega_t))E_t$. Since $\|E_t\| \leq 1$ and
$|\ell_t'(\trace(E_t \omega_t))| \leq L$, we have $\|M_t\| \leq 1$, as required in the analysis of the Matrix Multiplicative Weights algorithm. We have the following regret bound for the algorithm (proved in Appendix C):
\begin{theorem} \label{thm:main-mmw}
Setting $\eta = \sqrt{\frac{(\log 2)n}{4T}}$, the regret of the algorithm based on the update rule \eqref{eq:mmw-update} is bounded by $2L\sqrt{(\log 2)Tn}$.
\end{theorem}

\subsection{Proof of Theorem~\ref{thm1}}





Consider either the RFTL or MMW based online learning algorithm
described in the previous subsections, with the $1$-Lipschitz convex 
absolute loss function $\ell_t(x) = |x - b_t|$.
We run the algorithm in a sub-sequence of
the iterations, using only the measurements presented in those
iterations. 
The subsequence of iterations is determined as follows.
Let~$\omega_t$ denote the hypothesis maintained by the algorithm in 
iteration~$t$. We run the algorithm in iteration~$t$ if
$\ell_t(\trace(E_t \omega_t)) > \frac{2\varepsilon}{3}$.
Note that whenever $|\trace(E_t \omega_t) - \trace(E_t\rho)|
> \varepsilon$, we have  $\ell_t(\trace(E_t \omega_t))
> \frac{2\varepsilon}{3}$, so we
update the hypothesis according to the RFTL/MMW rule in that iteration.

As we explain next, the algorithm makes at most $\operatorname{O}(\tfrac{n}{\varepsilon^2})$ 
updates regardless of the number of measurements presented (i.e.,
regardless of the number of iterations),
giving the required mistake bound. For the true quantum state $\rho$, we have $\ell_t(\trace(E_t\rho)) < \frac{\varepsilon}{3}$ for all $t$. Thus if the algorithm makes $T$ updates
(i.e., we run the algorithm in~$T$ of the iterations), the regret bound implies that $\frac{2\varepsilon}{3}T \leq \frac{\varepsilon}{3}T + \operatorname{O}(\sqrt{Tn}\,)$. Simplifying, we get the bound $T = \operatorname{O}(\tfrac{n}{\varepsilon^2})$, as required.

\section{Learning Using Postselection} \label{sec:postselection}

In this section, we give a direct route to proving a slightly weaker version of Theorem \ref{thm1}: one that does not need the tools of convex optimization, but only tools intrinsic to quantum information.

In the following, by a ``register'' we mean a designated sequence of qubits.
Given a two-outcome measurement~$E$ on~$n$-qubits states, we define 
an operator~${\mathcal M}$ that ``postselects'' on acceptance by~$E$.
(While a measurement results in a random outcome distributed according
to the probability of acceptance or rejection, \emph{postselection\/} is a
hypothetical operation that produces an outcome of one's choice with
certainty.)
Let~$U$ be any unitary operation on~$n+1$ qubits
that maps states of the form~$|\psi\rangle |0\rangle $ to~$\sqrt{E}\, 
|\psi\rangle |0\rangle + \sqrt{\mathbbm{I} - E}\, |\psi\rangle |1\rangle $. 
Such a unitary operation always exists (see, e.g.,
\cite[Theorem~2.42]{Watrous18}). 
Denote the~$(n+1)$th qubit by register~$B$.
Let~$\Pi \coloneqq \mathbbm{I} \otimes |0\rangle\! \langle 0|$
be the orthogonal projection onto states that equal~$|0\rangle$ in 
register~$B$. Then we define the operator~${\mathcal M}$ as
\begin{equation}
{\mathcal M}(\varphi) \coloneqq
    \frac{1}{\trace(E\varphi)} \;
    \trace_B \!\left( U^{-1} \Pi U \left( \varphi \otimes |0\rangle\!
    \langle 0| \right) U^{-1} \Pi U \right) \enspace,
\end{equation}
if~$\trace(E\varphi) \neq 0$, and~${\mathcal M}(\varphi) \coloneqq 0$
otherwise. Here, $\trace_B$ is the \emph{partial trace\/} operator over
qubit~$B$~\cite[Section~1.1]{Watrous18}.
This operator~${\mathcal M}$ has the effect of mapping the quantum 
state~$\varphi$ to
the (normalized) post-measurement state when we perform the measurement~$E$
and get outcome~``yes'' (i.e., the measurement ``accepts'').
We emphasize that we use a fresh ancilla qubit initialized
to state~$|0\rangle$ as
register~$B$ in every application of the operator~$\mathcal{M}$.
We say that the postselection succeeds with probability~$\trace(E\varphi)$.

We need a slight variant of a well-known result, which Aaronson called the \textquotedblleft Quantum Union Bound\textquotedblright\ (see, for example, \cite{aar:qmaqpoly,aarbados,wilde}).
\begin{theorem}
[variant of Quantum Union Bound; \cite{Gao15}]\label{qub}
Suppose we have a sequence of two-outcome measurements $E_{1},\ldots,E_{k}$,
such that each $E_{i}$\ accepts a certain mixed state $\varphi$\ with 
probability at least $1-\varepsilon$. \
Consider the corresponding operators~${\mathcal M}_1, {\mathcal M}_2,
\dotsc, {\mathcal M}_k$ that postselect on acceptance by the respective 
measurements~$E_1, E_2, \dotsc, E_k$.
Let~$\widetilde{\varphi}$ denote the state~$({\mathcal M}_k {\mathcal M}_{k-1}
\dotsb {\mathcal M}_1)(\varphi)$ obtained by applying each of the~$k$ 
postselection operations in succession.
Then the probability that all
the postselection operations succeed, i.e., the~$k$ measurements all
accept~$\varphi$, is at least~$1 - 2\sqrt{k\varepsilon}$. Moreover,
$\left\Vert \widetilde{\varphi
}-\varphi\right\Vert _{\trace}\leq 4 \sqrt{k\varepsilon}$.
\end{theorem}
We may infer the above theorem by applying Theorem~1 from~(\cite{Gao15})
to the state~$\varphi$ augmented with~$k$ ancillary qubits~$B_1, B_2,
\dotsc, B_k$ initialized to~$0$, and considering~$k$ orthogonal projection 
operators~$U^{-1}_i \Pi_i U_i$, where the unitary operator~$U_i$ and the
projection operator~$\Pi_i$ are as defined for the postselection 
operation~$\mathcal{M}_i$ for~$E_i$. The~$i$th projection
operator~$U^{-1}_i \Pi_i U_i$ acts on the registers holding~$\varphi$ and 
the~$i$th ancillary qubit~$B_i$.

We prove the main result of this section using suitably defined 
postselection operators in an online learning algorithm (proof in Appendix D):
\begin{theorem} \label{thm:postselection-main}
Let $\rho$\ be an unknown $n$-qubit mixed state, let $E_{1},E_{2},\ldots$\ be
a sequence of two-outcome measurements, and let $\varepsilon>0$. \ There exists a
strategy for outputting hypothesis states $\omega_{0},\omega_{1},\ldots$,
where $\omega_{t}$\ depends only on $E_{1},\ldots,E_{t}$ and real numbers
$b_{1},\ldots,b_{t}$ in~$[0,1]$, such that as long as $\left\vert b
_{t}-\trace \!\left(  E_{t}\rho\right)  \right\vert \leq\varepsilon/3$
for every $t$, we have%
\[
\left\vert \trace \!\left(  E_{t+1}\omega_{t}\right)
-\trace \!\left(  E_{t+1}\rho\right)  \right\vert >\varepsilon
\]
for at most $\operatorname{O}\!\!\left(  \frac{n}{\varepsilon^{3}}\log\frac{n}{\varepsilon
}\right)  $\ values of $t$. \ Here the $E_{t}$'s and $b_{t}$'s can
otherwise be chosen adversarially.
\end{theorem}

\section{Learning Using Sequential Fat-Shattering Dimension} \label{sec:dimension}

In this section, we prove regret bounds using the notion of \emph{sequential
fat-shattering dimension\/}.
Let $S$ be a set of functions $f : U \rightarrow [0,1]$,
and~$\varepsilon > 0$. \ Then, following
\cite{rst}, let the
$\varepsilon$-\textit{sequential fat-shattering dimension} of $S$, or
$\operatorname*{sfat}_{\varepsilon}(S)$, be the largest $k$ for which
we can construct a complete binary tree $T$ of depth $k$, such that

\begin{itemize}
\item each internal vertex $v \in T$ has associated with it a point
$x_v \in U$ and a real $a_v \in [0,1]$, and

\item for each leaf vertex $v \in T$ there exists an $f \in S$ that
causes us to reach $v$ if we traverse $T$ from the root such that at 
any internal node~$w$ we traverse the left subtree if $f(x_w) \le a_w 
- \varepsilon$ and the right subtree if $f(x_w) \ge a_w + \varepsilon$.
If we view the leaf~$v$ as a~$k$-bit string, the function~$f$ is such 
that for all ancestors~$u$ of~$v$, we have~$f(x_u) \le a_u - \varepsilon$
if~$v_i = 0$, and~$f(x_u) \ge a_u + \varepsilon$ if~$v_i = 1$, when~$u$
is at depth~$i-1$ from the root.
\end{itemize}
An~$n$-qubit state~$\rho$ induces a function~$f$ on the set of two-outcome
measurements~$E$ defined as~$f(E) \coloneqq \trace(E\rho)$. With this
correspondence in mind, we establish a bound on the sequential 
fat-shattering dimension of the set of~$n$-qubit quantum states.
The bound is based on a generalization of \textquotedblleft random access
coding\textquotedblright\ (\cite{Nayak99,antv}) called ``serial
encoding''. We derive the
following bound on the length of serial encoding. Let~$\mathrm{H}(x)
\coloneqq -x \log_2 x - (1-x) \log_2 (1-x)$ be the binary entropy function.
\begin{corollary}
\label{cor-se}
Let~$k$ and~$n$ be positive integers. For each~$k$-bit
string~$y \coloneqq y_{1}\dotsb y_{k}$, let~$\rho_{y}$ be an~$n$-qubit
mixed state such that for each~$i \in \{1, 2, \dotsc, k \}$, there is a
two-outcome measurement~$E'$ that depends only on~$i$ and the 
prefix~$v \coloneqq y_1 y_2 \dotsb y_{i-1}$, and has the following properties
\begin{enumerate}

\item[(iii)] if~$y_{i} = 0$ then~${\mathrm{Tr}}(E^{\prime}\rho_{y}) \le a_{v}
- \varepsilon$, and

\item[(iv)] if~$y_{i} = 1$ then~${\mathrm{Tr}}(E^{\prime}\rho_{y}) \ge a_{v} +
\varepsilon$,
\end{enumerate}
where~$\varepsilon \in (0,1/2]$ and~$a_{v}\in [0, 1 ]$ is a \textquotedblleft pivot
point\textquotedblright\ associated with the prefix~$v$. Then
\[
n \quad \geq \quad \left( 1- \mathrm{H} \left(
    \frac{1-\varepsilon}{2} \right) \right) k \enspace.
\]
In particular, $k = \operatorname{O} \!\left(n / \varepsilon^2 \right)$.
\end{corollary}
(The proof is presented in Appendix E).

%
Corollary~\ref{cor-se} immediately implies the following theorem:
\begin{theorem}
\label{sfatthm} Let $U$ be the set of two-outcome measurements $E$ on
an $n$-qubit state, and let $S$ be the set of all functions $f : U
\rightarrow [0,1]$ that have the form $f(E) \coloneqq \trace(E \rho)$ for some
$\rho$. \ Then for all $\varepsilon > 0$, we have
$\operatorname*{sfat}_{\varepsilon}(S) = \operatorname{O}
\! \left(n / \varepsilon^2 \right)$. 
\end{theorem}
Theorem \ref{sfatthm} strengthens an earlier result due to
\cite{aar:learn}, which proved the same upper bound for the
``ordinary'' (non-sequential) fat-shattering dimension of quantum
states considered as a hypothesis class.

Now we may use existing results from the literature, which relate
sequential fat-shattering dimension to online learnability. \ In
particular, in the non-realizable case, \cite{rst} recently showed the
following:
\begin{theorem}
[\cite{rst}] \label{rstthm}
Let $S$ be a set of functions $f : U \rightarrow [0,1]$ and for every
integer~$t \ge 1$, let~$\ell_t : [0,1] \rightarrow \reals$ be a convex,
$L$-Lipschitz loss function. \ Suppose we
are sequentially presented elements $x_1,x_2,\ldots \in U$, with each
$x_t$ followed by the loss function~$\ell_t$. \ Then there
exists a learning strategy that lets us output a sequence of
hypotheses $f_1,f_2,\ldots \in S$, such that the regret is
upper-bounded as:
$$ \sum_{t=1}^{T} \ell_t\left( f_t(x_t) \right) \le
\min_{f \in S} \sum_{t=1}^{T} \ell_t \left( f(x_t) \right) + 2 L T ~
\inf_{\alpha} \left\{ 4 \alpha + \frac{12}{\sqrt{T}} \int_{\alpha}^{1}
\sqrt{ \operatorname{sfat}_{\beta}(S) \log\left( \frac{2\mathrm{e}T}{\beta}
\right)} \mathrm{d}\beta \right\}. $$
\end{theorem}
This follows from Theorem~8 in~(\cite{rst}) as in the proof of
Proposition~9 in the same article.

Combining Theorem \ref{sfatthm} with Theorem \ref{rstthm} gives us
the following:
\begin{corollary}
\label{sfatcor} Suppose we are presented with a sequence of two-outcome
measurements $E_1,E_2,\ldots$ of an $n$-qubit state, with each $E_t$
followed by a loss function~$\ell_t$ as in Theorem~\ref{rstthm}.
\ Then there exists a
learning strategy that lets us output a sequence of hypothesis states
$\omega_1,\omega_2,\ldots$ such that the regret after the first
$T$ iterations is upper-bounded as:
$$ \sum_{t=1}^{T} \ell_t \left( \trace(E_t \omega_t) \right) \le
\min_{\omega \in C_n} \sum_{t=1}^{T} \ell_t \left( \trace(E_t \omega) \right) +
\operatorname{O}\! \left( L \sqrt{n T} \log^{3/2} T \right). $$
\end{corollary}

Note that the result due to \cite{rst} is non-explicit. \ In other words,
by following this approach, we do not derive any specific online
learning algorithm for quantum states that has the stated upper bound 
on regret; we only prove non-constructively that such an algorithm exists. 

We expect that the approach in this section, based on sequential
fat-shattering dimension, could also be used to prove a mistake bound
for the realizable case, but we leave that to future work.

\section{Open Problems}

We conclude with some questions arising from this work. The regret bound
established in Theorem \ref{thm2} for~$L_1$ loss is tight.
Can we similarly achieve optimal regret for other loss
functions of interest, for example for~$L_2$-loss? It would also be interesting to obtain regret bounds in terms of the loss of the best quantum state in hindsight, as opposed to $T$
(the number of iterations), using the techniques in this article. Such a
bound has been shown by~\cite[Lemma~3.2]{tsuda} for~$L_2$-loss using the
Matrix Exponentiated Gradient method.

In what cases can one do online learning of quantum states, not only with few samples, but also with a polynomial amount of computation?
What is the tight generalization of our results to measurements with~$d$ outcomes?
Is it the case, in online learning of quantum states, that {\em any} algorithm works, so long as it produces hypothesis states that are approximately consistent with all the data seen so far? \ Note that none of our three proof techniques seem to imply this general conclusion.
\bibliographystyle{plainnat}
\bibliography{thesis}
\newpage
\appendix  \label{sec:appendix1}
\section{Auxiliary Lemmas}

The following lemma is from  (\cite{tsuda}), given here for completeness.
\begin{lemma}\label{auxlemma}
For Hermitian matrices $A, B$ and Hermitian PSD matrix $X$, if $A\succeq B$, then $\trace(AX) \ge \trace(BX)$.
\end{lemma}
\begin{proof}
Let $C \coloneqq A-B$. By definition, $C\succeq 0$. It suffices to show that $\trace(CX) \ge 0$. Let $VQV^\dagger$ be the eigen-decomposition of $X$, and let $C = VPV^\dagger$, where $P \coloneqq  V^\dagger C V \succeq 0$. Then $\trace(CX) = \trace(VPQV^\dagger) = \trace(PQ) = \sum_{i=1}^n P_{ii}Q_{ii}.$ Since $P\succeq 0$ and all the eigenvalues of $X$ are nonnegative, $P_{ii}\ge 0$, $Q_{ii}\ge 0$. Therefore $\trace(CX)\ge 0$.
\end{proof}
\begin{lemma}
If $A, B$ are Hermitian matrices, then $\trace(AB)\in\reals$.
\end{lemma}
\begin{proof}
The proof is similar to Lemma \ref{auxlemma}. Let $VQV^\dagger$ be the eigendecomposition of $A$. Then $Q$ is a real diagonal matrix. We
have $B = VPV^\dagger$, where $P \coloneqq V^\dagger B V$. Note that $P^\dagger = V^\dagger B^\dagger V = P$, so $P$ has a real diagonal. Then $\trace(AB) = \trace(VQV^\dagger VPV^\dagger) = \trace(VQPV^\dagger) = \trace(QP) = \sum_{i=1}^n Q_{ii}P_{ii}$. Since $Q_{ii}, P_{ii}\in\reals$ for all $i$, $\trace(AB) \in\reals$.
\end{proof}

\section{Proof of Theorem \ref{thm:main}}
\label{sec:rftl-main-proof}

\begin{proof}[Proof of Theorem \ref{thm:main}]
Since $\ell_t$ is convex, for all $ \varphi \in\K$,
\begin{align*}
\ell_t(\trace(E_t \omega_t)) - \ell_t(\trace(E_t \varphi))  \le \ell_t'(\trace(E_t
\omega_t )) \left[ \trace(E_t \omega_t ) - \trace(E_t \varphi) \right]
= \nabla_t \bullet \left( \omega_t - \varphi \right) \enspace.
\end{align*} 
(Recall that~`$\bullet$' denotes the trace inner-product between 
complex matrices of the same dimensions.)
Summing over $t$,
\begin{align*}
\sum_{t=1}^T [\ell_t(\trace(E_t \omega_t)) - \ell_t(\trace(E_t \varphi))] &\le
\sum_{t=1}^T \left[ \trace(\nabla_t  \omega_t) - \trace(\nabla_t \varphi)
\right]
\enspace.
\end{align*}

Define $g_t(X) = \nabla_t\bullet X$, and $g_0(X) = \frac{1}{\eta} R(X)$, where $R(X)$ is the negative von Neumann Entropy
of~$X$ (in nats). Denote $D_R^2 \coloneqq \max_{\varphi, \varphi'\in \K} \{R(\varphi) - R(\varphi')\}$. By
\cite[Lemma 5.2]{oco}, for any $\varphi\in\K$, we have
\begin{equation}\label{lem5.2}
\sum_{t=1}^T [g_t( \omega_t) - g_t(\varphi)] \le \sum_{t=1}^T \nabla_t\bullet(
\omega_t -  \omega_{t+1}) + \frac{1}{\eta} D_R^2
\enspace.
\end{equation}

Define $\Phi_t(X) = \{ \eta \sum_{s=1}^t \nabla_s \bullet X + R(X)\}$,
then the convex program in line 5 of Algorithm \ref{alg:rftl} finds the
minimizer of $\Phi_t(X)$ in $\K$. The following claim shows that that
the minimizer is always \emph{positive definite\/} (proof provided later in this section):

\begin{claim}\label{claim2}
For all $t\in \{1, 2, ..., T\}$, we have $ \omega_t\succ 0$.
\end{claim}
For $X\succ 0$, we can write $R(X) = \trace(X\log X)$, and define
\[
\nabla \Phi_t(X) \coloneqq \eta \sum_{s = 1}^t \nabla_s + \mathbbm{I} + \log X
\enspace.
\]
The definition of $\nabla \Phi_t(X)$ is analogous to the gradient of $\Phi_t(X)$ if the function is defined over real symmetric matrices. Moreover, the following condition, similar to the optimality condition over a real domain, is satisfied (proof provided later in this section).
\begin{claim} \label{claim_opt}
For all~$t \in \{1, 2, \dotsc, T-1\}$,
\begin{align}
\nabla \Phi_t( \omega_{t+1})\bullet ( \omega_t- \omega_{t+1}) \ge 0 \enspace.
\end{align}
\end{claim}

Denote $$B_{\Phi_t}( \omega_t\| \omega_{t+1}) \coloneqq \Phi_t(
\omega_t) - \Phi_t( \omega_{t+1}) - \nabla \Phi_t( \omega_{t+1})\bullet
( \omega_t- \omega_{t+1})
\enspace.$$
Then by the Pinsker inequality (see, for example, \cite{pinsker} and the references therein),
\begin{align*}
\frac{1}{2} \| \omega_t -  \omega_{t+1}\|_{\trace}^2 &\le  \trace(
\omega_t\log  \omega_t) - \trace( \omega_t\log  \omega_{t+1}) =
B_{\Phi_t}( \omega_t\| \omega_{t+1})
\enspace.
\end{align*}
We have
\begin{align}
B_{\Phi_t}( \omega_t\| \omega_{t+1}) & =\Phi_t( \omega_t) - \Phi_t(
\omega_{t+1}) - \nabla \Phi_t( \omega_{t+1})\bullet ( \omega_t- \omega_{t+1}) \nonumber\\
 &\le \Phi_t( \omega_t) - \Phi_t( \omega_{t+1}) \nonumber\\
 &= \Phi_{t-1}( \omega_t) - \Phi_{t-1}( \omega_{t+1}) +\eta \nabla_t
\bullet( \omega_t- \omega_{t+1})\nonumber\\
 &\le \eta \nabla_t \bullet ( \omega_t- \omega_{t+1}) \enspace, 
\end{align}
where the first inequality follows from Claim~\ref{claim_opt}, and the second
because $\Phi_{t-1}( \omega_t) \le \Phi_{t-1}( \omega_{t+1})$ ($ \omega_t$
minimizes~$\Phi_{t-1}(X)$).  Therefore
\begin{equation}\label{eq1}
\frac{1}{2} \| \omega_t -  \omega_{t+1}\|_{\trace}^2 \le  \eta \nabla_t
\bullet ( \omega_t- \omega_{t+1})
\enspace.
\end{equation}
Let~$\| M \|_{\trace}^*$ denote the dual of the trace norm, i.e.,
the spectral norm of the matrix~$M$.
By Generalized Cauchy-Schwartz~\cite[Exercise IV.1.14, page~90]{Bhatia97},
\begin{align*}
\nabla_t \bullet ( \omega_t- \omega_{t+1}) &\le \|\nabla_t\|_{\trace}^*
\; \| \omega_t- \omega_{t+1}\|_{\trace}\\
&\le \|\nabla_t\|_{\trace}^* \sqrt{2\eta\nabla_t \bullet ( \omega_t-
\omega_{t+1}) }
\enspace. & \mbox{by\ Eq.~\eqref{eq1}.}
\end{align*}
Rearranging,
$$
\nabla_t \bullet ( \omega_t- \omega_{t+1})\le 2\eta \|\nabla_t\|_{\trace}^{*2}\le 2\eta G_R^2
\enspace,
$$
where $G_R$ is an upper bound on $\|\nabla_t\|_{\trace}^*$.
Combining with Eq.~\eqref{lem5.2}, we arrive at the following bound
\begin{align*}
\sum_{t=1}^T \nabla_t\bullet( \omega_t- \varphi) \le \sum_{t=1}^T \nabla_t
\bullet ( \omega_t -  \omega_{t+1}) + \frac{1}{\eta} D_R^2
\le 2\eta TG_R^2 + \frac{1}{\eta}D_R^2 \enspace.
\end{align*}
Taking $\eta = \frac{D_R}{G_R\sqrt{2T}}$, we get
$
\sum_{t=1}^T \nabla_t\bullet( \omega_t- \varphi) \le 2D_RG_R\sqrt{2T}
$.
Going back to the regret bound,
$$
\sum_{t=1}^T [\ell_t(\trace(E_t \omega_t)) - \ell_t(\trace(E_t \varphi))] \le
\sum_{t=1}^T \nabla_t\bullet( \omega_t- \varphi) \le 2D_RG_R\sqrt{2T}
\enspace.
$$
We proceed to show that~$D_R =\sqrt{(\log 2)n}$. Let~$\Delta_{2^n}$ denote the
set of probability distributions over~$[2^n]$. By definition,
$$
D_R^2 = \max_{ \varphi, \varphi'\in\K} \{R(\varphi)-R(\varphi')\} = \max_{\varphi \in\K} -R(\varphi ) = \max_{\lambda\in\triangle_{2^n}} \sum_{i=1}^{2^n} \lambda_i\log\frac{1}{\lambda_i}
= n \log 2 
\enspace.
$$

Since the dual norm of the trace norm is the spectral norm, we have
\begin{align*}
\|\nabla_t\|_{\trace}^* &= \|\ell_t'(\trace(E_t \omega_t))E_t\|
\le L\|E_t\|\le L \enspace.
\end{align*}
Therefore $\sum_{t=1}^T [(\ell_t(\trace(E_t \omega_t)) - \ell_t(\trace(E_t
\varphi))] \le 2L\sqrt{(2 \log 2) nT\/}$.
\end{proof}


\begin{proof}[Proof of Claim~\ref{claim2}]
Let $P\in\K$ be such that $\lambda_{\min}(P) = 0$. Suppose $P=VQV^\dagger$, where $Q$ is a diagonal matrix with real values on the diagonal. Assume
that~$Q_{1,1} = \lambda_{\max}(P)$ and~$Q_{2^n,2^n} = \lambda_{\min}(P)
= 0$. Let $P' = VQ'V^\dagger$ such that $Q'_{1,1} = Q_{1,1}-\varepsilon$, $Q'_{2^n, 2^n} = \varepsilon$ for $\varepsilon < \lambda_{\max}(P)$, and $Q'_{ii} = Q_{ii}$ for $i\in
\{2, 3, ..., 2^n-1\}$, so $P'\in\K$. We show that there exists $\varepsilon > 0$ such that $\Phi_t(P')\le \Phi_t(P)$. Expanding both sides of the inequality,
we see that it is equivalent to showing that for some $\varepsilon$,
\begin{align*}
\eta \sum_{s=1}^t \nabla_s \bullet (P'-P) \le \lambda_1(P)\log\lambda_1(P) - \lambda_1(P')\log\lambda_1(P')-\varepsilon\log\varepsilon
\enspace.
\end{align*}
Let $\alpha=\lambda_1(P) = Q_{1,1}$, and $A=\eta \sum_{s=1}^t \nabla_s$.
The inequality then becomes
\begin{align*}
A \bullet (P'-P) \le \alpha\log \alpha - (\alpha-\varepsilon)\log(\alpha-\varepsilon)-\varepsilon\log\varepsilon
\enspace.
\end{align*}
Observe that $\left\Vert A\right\Vert  \le \eta \sum_{s=1}^t \left\Vert \nabla_s\right\Vert
= \eta \sum_{s=1}^t \left\Vert \ell'_s(\trace(E_s \omega_s)) E_s
\right\Vert \le \eta L t$.
So by the Generalized Cauchy-Schwartz inequality,
\begin{align*}
A \bullet (P'-P) &\le
\eta L t \left\Vert P'-P\right\Vert _{\trace} \le 2\varepsilon \eta L t
\enspace.
\end{align*}
Since $\eta, t, \alpha, L$ are finite and $-\log \varepsilon\rightarrow \infty$ as $\varepsilon\rightarrow 0$, there exists $\varepsilon$ small such that $2\eta
L t\le \log \alpha-\log\varepsilon$. We have
\begin{align*}
2\eta L t\varepsilon &\le \varepsilon\log \alpha - \varepsilon\log\varepsilon\\
&= \alpha\log \alpha - (\alpha-\varepsilon)\log \alpha -\varepsilon\log\varepsilon \\
&\le \alpha\log \alpha - (\alpha-\varepsilon)\log (\alpha-\varepsilon) -\varepsilon\log\varepsilon
\enspace.
\end{align*}
So there exists $\varepsilon > 0$ such that $\Phi_t(P') \le \Phi_t(P)$.
If $P$ has multiple eigenvalues that are 0, we can repeat the proof and
show that there exists a PD matrix $P'$ such that $\Phi_t(P')\le
\Phi_t(P)$. Since $ \omega_t$ is a minimizer of~$\Phi_{t-1}$ and $
\omega_1 \succ 0$, we conclude that $ \omega_t\succ 0$ for all $t$.
\end{proof}


\begin{proof}[Proof of Claim~\ref{claim_opt}.]
Suppose $\nabla \Phi_t( \omega_{t+1})\bullet ( \omega_t- \omega_{t+1}) <
0$. Let $a \in (0, 1)$ and $\bar{X} = (1-a)  \omega_{t+1} + a
\omega_t$, then $\bar{X}$ is a density matrix and is positive definite.
Define $\triangle = \bar{X} -  \omega_{t+1} = a( \omega_{t} -  \omega_{t+1})$. We have
\begin{align*}
\Phi_t(\bar{X}) - \Phi_t( \omega_{t+1}) &= a \nabla \Phi_t(
\omega_{t+1})\bullet ( \omega_{t} -  \omega_{t+1}) +
B_{\Phi_t}(\bar{X}\| \omega_{t+1})\\
&\le a \nabla \Phi_t( \omega_{t+1})\bullet ( \omega_{t} -  \omega_{t+1})
+ \frac{\trace(\triangle^2)}{\lambda_{\min}( \omega_{t+1})}\\&=a \nabla
\Phi_t( \omega_{t+1})\bullet ( \omega_{t} -  \omega_{t+1}) + \frac{a^2
\, \trace(( \omega_{t} -  \omega_{t+1})^2)}{\lambda_{\min}( \omega_{t+1})} \enspace.
\end{align*}
The above inequality is due to \cite[Theorem 2]{audenaert}. Dividing by $a$ on both sides, we have
\begin{align*}
\frac{\Phi_t(\bar{X}) - \Phi_t( \omega_{t+1})}{a} &\le \nabla \Phi_t(
\omega_{t+1})\bullet ( \omega_{t} -  \omega_{t+1}) + \frac{a\trace((
\omega_{t} -  \omega_{t+1})^2)}{\lambda_{\min}( \omega_{t+1})}
\enspace.
\end{align*}
So we can find $a$ small enough such that the right hand side of the
above inequality is negative. However, we would have $\Phi_t(\bar{X}) -
\Phi_t( \omega_{t+1}) < 0$, which is a contradiction. So $\nabla \Phi_t(
\omega_{t+1})\bullet ( \omega_t- \omega_{t+1}) \ge 0$.
\end{proof}

\section{Proof of Theorem~\ref{thm:main-mmw}}
\label{sec:mmw-main-proof}

\begin{proof}[Proof of Theorem~\ref{thm:main-mmw}] Note that for any density matrix $\varphi$, we have $M_t \bullet
\varphi = \frac{1}{L}\ell_t'(\trace(E_t \omega_t))\trace(E_t\varphi)$. Then, the regret bound for Matrix Multiplicative Weights \cite[Theorem 3.1]{arora-kale} implies that for any density matrix $\varphi$, we have
\[\sum_{t=1}^T \ell_t'(\trace(E_t \omega_t))\trace(E_t \omega_t) \leq
\sum_{t=1}^T \ell_t'(\trace(E_t \omega_t))\trace(E_t\varphi) + \eta LT + \frac{L\log(2^n)}{\eta}
\enspace. \]
Here, we used the bound $M_t^2 \bullet  \omega_t \leq 1$. Next, since $\ell_t$ is convex, we have 
\[\ell_t'(\trace(E_t \omega_t))\trace(E_t \omega_t) - \ell_t'(\trace(E_t
\omega_t))\trace(E_t\varphi) \geq \ell_t(\trace(E_t \omega_t)) - \ell_t(\trace(E_t\varphi))
\enspace.\]
Using this bound, and the stated value of $\eta$, we get the required regret bound.
\end{proof}

\section{Proof of Theorem~\ref{thm:postselection-main}}
\label{sec:postselection-main-proof}

\begin{proof}[Proof of Theorem~\ref{thm:postselection-main}]
Let $\rho^{\ast}:=\rho^{\otimes k}$\ be an amplified version of $\rho$, 
over a Hilbert space of dimension $D \coloneqq 2^{kn}$, for some $k$ to be set
later. \ Throughout, we maintain a classical description of a
$D$-dimensional \textquotedblleft amplified hypothesis state\textquotedblright%
\ $\omega_{t}^{\ast}$, which we view as being the state of~$k$ registers
with~$n$ qubits each. We ensure that~$\omega_{t}^{\ast}$ is always symmetric 
under permuting the $k$
registers. \ Given $\omega_{t}^{\ast}$, our actual~$n$-qubit hypothesis state
$\omega_{t}$\ is then obtained by simply tracing out $k-1$\ of the registers.

Given an amplified hypothesis state $\omega^{\ast}$, let $E_{t}^{\ast}%
$\ be a two-outcome measurement\ that acts on $\omega^{\ast}$\ as follows: it
applies the measurement $E_{t}$\ to each of the $k$ registers separately, and
accepts if and only if the fraction of measurements that accept equals
$b_{t}$, up to an additive error at most $\varepsilon/2$.

Here is the learning strategy. \ Our initial hypothesis, $\omega_{0}^{\ast
}:=\mathbbm{I}/D$, is the $D$-dimensional maximally mixed state, corresponding to
$\omega_{0} \coloneqq \mathbbm{I}/2^{n}$. (The \emph{maximally mixed state\/}
corresponds to the notion of a uniformly random quantum superposition.)
\ For each $t\geq1$, we are given descriptions of the
measurements~$E_{1},\ldots,E_{t}$, as well as real numbers $b_{1}%
,\ldots,b_{t}$ in~$[0,1]$, such that $\left\vert b_{i}-\trace \left(
E_{i}\rho\right)  \right\vert \leq \varepsilon / 3$ for all $i\in\left[  t\right]  $. \ We
would like to update our old hypothesis $\omega_{t-1}^{\ast}$ to a new hypothesis
$\omega_{t}^{\ast}$, ideally such that the difference $\left\vert \trace \left(
E_{t+1}\omega_{t}\right)  -\trace \left(  E_{t+1}\rho\right)
\right\vert $\ is small.  We do so as follows:

\begin{itemize}
\item Given $b_{t}$, as well classical descriptions of $\omega_{t-1}%
^{\ast}$\ and $E_{t}$, decide whether $\trace \!\left(  E_{t}^{\ast
}\omega_{t-1}^{\ast}\right)  \geq1-\frac{\varepsilon}{6}$.

\item If yes, then set $\omega_{t}^{\ast}:=\omega_{t-1}^{\ast}$ (i.e.,
we do not change the hypothesis).

\item Otherwise, let $\omega_{t}^{\ast}$\ be the state obtained by applying
$E_{t}^{\ast}$\ to $\omega_{t-1}^{\ast}$ and postselecting on $E_{t}^{\ast}$\ accepting.
In other words, $\omega_{t}^{\ast} \coloneqq {\mathcal
M}(\omega_{t-1}^{\ast})$, where~${\mathcal M}$ is the operator
that postselects on acceptance by~$E_{t}^{\ast}$ (as defined above).
\end{itemize}

We now analyze this strategy. \ Call $t$ \textquotedblleft
good\textquotedblright\ if $\trace \!\left(  E_{t}^{\ast}\omega
_{t-1}^{\ast}\right)  \geq1-\frac{\varepsilon}{6}$, and \textquotedblleft
bad\textquotedblright\ otherwise. \ Below, we show that

\begin{enumerate}
\item[(i)] there are at most $\operatorname{O}\!\!\left(  \frac{n}{\varepsilon^{3}}\log
\frac{n}{\varepsilon}\right)  $\ bad $t$'s, and

\item[(ii)] for each good $t$, we have $\left\vert \trace \!\left(
E_{t}\omega_{t-1}\right)  - \trace(E_t \rho) \right\vert \leq\varepsilon$.
\end{enumerate}

We start with claim (i). \ Suppose there have been $\ell$\ bad $t$'s, call
them $t\!\left(  1\right)  ,\ldots,t\!\left(  \ell\right)  $, where $\ell
\leq\left(  n/\varepsilon\right)  ^{10}$\ (we justify this last assumption
later, with room to spare). \ Then there were $\ell$\ events where we 
postselected on $E_{t}^{\ast}$ accepting $\omega_{t-1}^{\ast}$.
We conduct a thought experiment, in which the learning strategy
maintains a quantum register initially in the maximally mixed
state~$\mathbbm{I}/D$, and applies the postselection operator
corresponding to~$E_t^\ast$ to the quantum register whenever~$t$ is bad. \ Let
$p$ be the probability that all $\ell$ of these postselection events succeed. \ Then by
definition,%
\[
p=\trace \!\left(  E_{t\left(  1\right)  }^{\ast}\omega_{t\left(
1\right)  -1}^{\ast}\right)  \cdots\trace \!\left(  E_{t\left(
\ell\right)  }^{\ast}\omega_{t\left(  \ell\right)  -1}^{\ast}\right)
\leq\left(  1-\frac{\varepsilon}{6}\right)  ^{\ell}.
\]
On the other hand, suppose counterfactually that we had started with the
\textquotedblleft true\textquotedblright\ hypothesis, $\omega_{0}^{\ast}%
\coloneqq \rho^{\ast}=\rho^{\otimes k}$. \ In that case, we would have%
\begin{align*}
\trace \!\left(  E_{t\left(  i\right)  }^{\ast}\rho^{\ast}\right)   &
=\Pr\left[  E_{t\left(  i\right)  }\text{ accepts }\rho\text{\ between
}\left(  b_{t\left(  i\right)  }-\frac{\varepsilon}{2}\right)
k\text{\ and }\left(  b_{t\left(  i\right)  }+\frac{\varepsilon}%
{2}\right)  k\text{\ times}\right] \\
&  \geq1-2\operatorname{e}^{-2 k \left(  \varepsilon/6\right)^{2}}
\end{align*}
for all $i$. \ Here the second line follows from
the condition that $\left\vert \trace \!\left(  E_{t\left(  i\right)  }%
\rho\right)  -b_{t\left(  i\right)  }\right\vert \leq\varepsilon/6$,
together with the Hoeffding bound.

We now make the choice $k:=\frac{C}{\varepsilon^{2}}\log\frac{n}{\varepsilon}%
$, for some constant $C$ large enough that%
\[
\trace \!\left(  E_{t\left(  i\right)  }^{\ast}\rho^{\ast}\right)
\geq1-\frac{\varepsilon^{10}}{400n^{10}}%
\]
for all $i$. \ So by Theorem~\ref{qub}, all $\ell$\ postselection events would
succeed with probability at least%
\[
1-2 \sqrt{\ell \frac{\varepsilon^{10}}{400 n^{10}}}\geq 0.9 \enspace.
\]
We may write the maximally mixed state, $\mathbbm{I}/D$, as%
\[
\frac{1}{D}\rho^{\ast}+\left(  1-\frac{1}{D}\right)  \xi \enspace,
\]
for some other mixed state $\xi$. \ For this reason, even when we start with
initial hypothesis $\omega_{0}^{\ast}=\mathbbm{I}/D$, all $\ell$\ postselection events
still succeed with probability%
\[
p\geq\frac{0.9}{D} \enspace.
\]
Combining our upper and lower bounds on $p$ now yields%
\[
\frac{0.9}{2^{kn}}\leq\left(  1-\frac{\varepsilon}{6}\right)  ^{\ell}%
\]
or%
\[
\ell=\operatorname{O} \!\left(  \frac{kn}{\varepsilon}\right)  =\operatorname{O} \!\left(  \frac{n}{\varepsilon
^{3}}\log\frac{n}{\varepsilon}\right)  ,
\]
which incidentally justifies our earlier assumption that $\ell\leq\left(
n/\varepsilon\right)  ^{10}$.

It remains only to prove claim (ii). \ Suppose that%
\begin{equation}
\label{eq-case-ii}
\trace  \!\left(  E_{t}^{\ast}\omega_{t-1}^{\ast}\right)  \geq
1-\frac{\varepsilon}{6} \enspace.
\end{equation}
Imagine measuring~$k$ quantum registers prepared in the joint
state~$\omega_{t-1}^\ast$, by applying~$E_t$ to each register. Since the
state~$\omega_{t-1}^\ast$ is symmetric under permutation of the~$k$
registers, we have that~$\trace(E_t \omega_{t-1})$, the probability 
that~$E_t$ accepts the first register, equals the expected fraction
of the~$k$ registers that~$E_t$ accepts.
The bound in Eq.~\eqref{eq-case-ii} means that, with probability at 
least $1-\frac{\varepsilon}{6}$ over the measurement outcomes, the
fraction of registers which~$E_{t}$ accepts is within
$\pm \varepsilon/2$\ of $b_{t}$. \ 
The~$k$ measurement outcomes are not necessarily independent, but the fraction of registers accepted 
never differs from $b_{t}$ by more than~$1$. \ So by the union bound, we have
\[
\left\vert \trace \!\left(  E_{t}\omega_{t-1}\right)  -b
_{t}\right\vert \leq\frac{\varepsilon}{2}+\frac{\varepsilon}{6}= \frac{2 \varepsilon}{3}
\enspace.
\]
Hence by the triangle inequality,
\[
\left\vert \trace \!\left(  E_{t}\omega_{t-1}\right)  - \trace(E_t \rho)
\right\vert \leq\frac{2 \varepsilon}{3} + \left| b_t - \trace(E_t \rho) \right| \leq \varepsilon
\enspace,
\]
as claimed.
\end{proof}

\section{Proof of Corollary~\ref{cor-se}}
\label{sec:cor-se-proof}

We begin with a bound for a generalization of \textquotedblleft random
access
coding\textquotedblright\ (\cite{Nayak99,antv}) or what is also known 
as the Index function
problem in communication complexity. \ The generalization was called
\textquotedblleft serial encoding\textquotedblright\
by~\cite{Nayak99} and arose in the context of quantum finite automata.
\ The serial encoding problem is also called
Augmented Index in the literature on streaming algorithms.

The following theorem places a bound on how few qubits serial encoding
may use. In other words, it
bounds the number of bits we may encode in an~$n$-qubit quantum state
when an arbitrary bit out of the~$n$ may be recovered well via a
two-outcome measurement. The bound holds even when the measurement for
recovering~$y_i$ may depend adaptively on the previous bits~$y_1 y_2
\dotsb y_{i-1}$ of~$y$, \emph{which we need not know\/}.

\begin{theorem}[\cite{Nayak99}]
\label{thm-se}
Let~$k$ and~$n$ be positive integers. For each~$k$-bit
string~$y \coloneqq y_{1}\dotsb y_{k}$, let~$\rho_{y}$ be an~$n$-qubit
mixed
state such that for each~$i \in \{1, 2, \dotsc, k \}$, there is a
two-outcome measurement~$E$ that depends only on~$i$ and the prefix~$y_1
y_2
\dotsb y_{i-1}$, and has the following properties
\begin{enumerate}
\item[(i)] if~$y_{i} = 0$ then~${\mathrm{Tr}}(E \rho_{y}) \le p$, and

\item[(ii)] if~$y_{i} = 1$ then~${\mathrm{Tr}}(E\rho_{y}) \ge 1 - p$,
\end{enumerate}
\noindent where~$p \in [0,1/2]$ is the error in predicting
the bit~$y_i$ at vertex~$v$.
(We say~$\rho_y$ \textquotedblleft serially
encodes\textquotedblright~$y$.)
\ Then~$n\geq(1-\mathrm{H}(p))k$.
\end{theorem}

In Appendix~\ref{sec-rac}, we present a strengthening of this bound
when the bits of~$y$ may be only be recovered in an adaptive order that
is \textit{a priori\/} unknown. The stronger bound may be of independent
interest.

In the context of online learning, the measurements used in recovering
bits
from a serial encoding are required to predict the bits with probability
bounded away from given ``pivot points''. Theorem~\ref{thm-se} may be
specialized to this case as in Corollary~\ref{cor-se}, which we prove
below.

\begin{proof}[Proof of Corollary~\ref{cor-se}]
\ This is a consequence of Theorem~\ref{thm-se}, when combined with 
the following observation. \ Given the
measurement operator~$E^{\prime}$, parameter~$\varepsilon$, and
pivot point~$a_{v}$ as in the statement of the corollary, we define 
a new two-outcome measurement~$E$ to be associated with vertex~$v$:
\[
E\quad:=\quad%
\begin{cases}
\frac{E^{\prime}}{2a_{v}} & \text{ if }a_{v}\geq\frac{1}{2}
    \enspace, \quad \text{ and } \vspace{1.5ex} \\
\frac{1}{2(1 - a_v)} \left( E^{\prime} + (1 - 2 a_{v}) \mathbb{I} \right)
     & \text{ if }a_{v}<\frac{1}{2} \enspace.
\end{cases}
\]
The measurement~$E$ may be interpreted as producing a fixed 
outcome~$0$ or~$1$ with some probability depending on~$a_{v}$,
and applying the given measurement~$E^{\prime}$ with the remaining 
probability, so as to translate the pivot point~$a_{v}$ to~$1/2$.

We may verify that the operator~$E$ satisfies the requirements~(i)
and~(ii) of
Theorem~\ref{thm-se} with~$p \coloneqq (1-\varepsilon)/2$.
We therefore conclude that~$n \ge (1 -
\mathrm{H}((1 - \varepsilon)/2) k$. Since~H$(1/2 - \delta) \le 1 - 2
\delta^2$, for~$\delta \in [0,1/2]$, we get~$k = \operatorname{O}
\!\left( n / \varepsilon^2 \right)$.
\end{proof}

\section{Lower bound on quantum random access codes}
\label{sec-rac}

Here we present an alternative proof of the linear lower bound on
quantum random access codes~\cite{Nayak99,antv}. It goes via the Matrix
Multiplicative Weights algorithm, but gives us a slightly weaker 
dependence on decoding error.  We also present an extension of the 
original bound to more general codes. These may be of independent interest.

\begin{theorem}
Let $k$ and $n$ be positive integers with $k > n$. For all $k$-bit strings $y = y_1, y_2, \ldots, y_k$, let $\rho_y$ be the $n$-qubit quantum mixed state that encodes $y$. Let $p \in [0, 1/2]$ be an error tolerance parameter.
Suppose that there exist measurements $E_1, E_2, \ldots, E_k$ such that for all $y \in \{0, 1\}^k$ and all $i \in [k]$, we have $|\trace(E_i\rho_y) - y_i| \leq p$. Then $n \geq \frac{(1/2-p)^2}{4(\log 2)}k$.
\end{theorem}
\begin{proof}
Run the MMW algorithm described in Section~\ref{sec-mmw} with the
absolute loss function~$\ell_t(x) \coloneqq |x - y_t|$ for $t = 1, 2, \ldots, k$ iterations. In iteration~$t$, provide as feedback $E_t$ and the label $y_t \in \{0, 1\}$ defined as follows:
\[ y_t = \begin{cases}
	0 & \text{ if } \trace(E_t \omega_t) > \frac{1}{2} \\
	1 & \text{ if } \trace(E_t \omega_t) \leq \frac{1}{2} \enspace.
\end{cases}\]
Let $y \in \{0, 1\}^k$ be the bit string formed at the end of the process. Then it is easy to check the following two properties by the construction of the labels: for any $t \in [k]$, we have 
\begin{enumerate}
	\item $\ell_t( \omega_t) = |\trace(E_t \omega_t) - y_t| \geq 1/2$, and

	\item $\ell_t(\rho_y)) = |\trace(E_t\rho_y) - y_t| \leq p$.
\end{enumerate}
By Theorem~\ref{thm:main-mmw}, the MMW algorithm with absolute loss has
a regret bound of $2\sqrt{(\log 2)kn}$. So the above bounds imply that $k/2 \leq pk + 2\sqrt{(\log 2)kn}$, which implies that $n \geq \frac{(1/2-p)^2}{4\log 2}k$.
\end{proof}
Note that in the above proof, we may allow the measurement in the~$i$th
iteration, i.e., the one used to decode the~$i$th bit, to depend on the
previous bits~$y_1, y_2, \dotsc, y_{i-1}$. Thus, the lower bound also
applies to serial encoding.

Next we consider encoding of bit-strings~$y$ into quantum states~$\rho_y$ 
with a more relaxed notion of decoding. The encoding is such that given
the encoding for an unknown string~$y$, \emph{some\/} bit~$i_1$ of~$y$ can be
decoded. Given the value~$y_{i_1}$ of of this bit, a new bit~$i_2$ of~$y$ 
can be decoded, and the index~$i_2$ may depend on~$y_{i_1}$. More
generally, given a sequence of bits~$y_{i_1} y_{i_2} \dotsc y_{i_j}$ 
that may be decoded in this manner, a new bit~$i_{j+1}$ of~$y$ can be
decoded, for any~$j \in \{0, 1, \dotsc, k-1 \}$. Here, the index~$i_{j+1}$
and the measurement used to recover the corresponding bit of~$y$ may
depend on the sequence of bits~$y_{i_1} y_{i_2} \dotsc y_{i_j}$.
We show that \emph{even with
this relaxed notion of decoding\/}, we cannot encode more than a linear 
number of bits into an~$n$-qubit state.

We first formalize the above generalization of random access
encoding.
We view a complete binary tree of depth~$d\geq0$ as consisting of
vertices~$v\in\left\{  0,1\right\}  ^{\leq d}$. \ The root of the tree is
labeled by the empty string~$\epsilon$ and each internal vertex~$v$ of the
tree has two children~$v0,v1$. We specify an adaptive
sequence of measurements through a \textquotedblleft measurement decision
tree\textquotedblright. \ The tree specifies the measurement to be applied
next, given a prefix of such measurements along with the corresponding outcomes.

\begin{definition}
Let~$k$ be a positive integer. \ A \emph{measurement decision tree\/} of
depth~$k$ is a complete binary tree of depth~$k$, each internal vertex~$v$ of
which is labeled by a triple~$(S,i,E)$, where~$S\in\left\{  1,\dotsc
,k\right\}  ^{l}$ is a sequence of length~$l:=\left\vert v\right\vert $ of
distinct indices, $i\in\left\{  1,\dotsc,k\right\}  $ is an index that does
not occur in~$S$, and~$E$ is a two-outcome measurement. \ The sequences
associated with the children~$v0,v1$ of~$v$ (if defined) are both equal
to~$(S,i)$.
\end{definition}

For a~$k$-bit string~$y$, and sequence~$S \coloneqq (i_{1},i_{2}, \dotsc, 
i_{l})$ with~$0\leq l\leq k$ and~$i_{j}\in\left\{  1,2,\dotsc,k\right\} $,
let~$y_{S}$ denote the substring~$y_{i_{1}}y_{i_{2}}\dotsb y_{i_{l}}$.

\begin{theorem}
\label{thm-arac} Let~$k$ and~$n$ be positive integers. For each~$k$-bit
string~$y \coloneqq y_{1}\dotsb y_{k}$, let~$\rho_{y}$ be an~$n$-qubit mixed 
state (we say~$\rho_y$ 
\textquotedblleft encodes\textquotedblright~$y$). \ Suppose there exists a
measurement decision tree~$T$ of depth~$k$ such that for each internal
vertex~$v$ of~$T$ and all~$y\in\left\{  0,1\right\}  ^{k}$ with~$y_{S}=v$,
where~$(S,i,E)$ is the triple associated with the vertex~$v$, we
have~$\left| {\mathrm{Tr}}(E \rho_{y}) - y_i \right| \le p_v$,
where~$p_{v}\in\lbrack0,1/2]$ is the error in predicting
the bit~$y_i$ at vertex~$v$. \ Then~$n\geq(1-\mathrm{H}(p))k$,
where~$\mathrm{H}$ is the binary entropy function, 
and~$p \coloneqq \tfrac{1}{k}
\sum_{l=1}^{k}\tfrac{1}{2^{l}}\sum_{v\in\left\{  0,1\right\}  ^{l}}p_{v}$
is the average error.
\end{theorem}

\begin{proof}
Let~$Y$ be a uniformly random~$k$-bit string. \ We define a random
permutation~$\Pi$ of~$\left\{  1,\ldots,k\right\}  $ correlated with~$Y$ that
is given by the sequence of measurements in the root to leaf path
corresponding to~$Y$. \ More formally, let~$\Pi(1):=i$, where~$i$ is the index
associated with the root of the measurement decision tree~$T$. For~$l\in
\left\{  2,\ldots,k\right\}  $, let~$\Pi(l):=j$, where~$j$ is the index
associated with the vertex~$Y_{\Pi(1)}Y_{\Pi(2)}\dotsb Y_{\Pi(l-1)}$ of the
tree~$T$. \ Let~$Q$ be a quantum register such that the joint state of~$YQ$ is
\[
\frac{1}{2^{k}}\sum_{y\in\left\{  0,1\right\}  ^{k}}\left\vert y\right\rangle\!
\left\langle y\right\vert \otimes\rho_{y} \enspace .
\]
The quantum mutual information between~$Y$ and~$Q$ is bounded
as~$\operatorname{I}(Y:Q)\leq\left\vert Q\right\vert =n$. \ Imagine having
performed the first~$l-1$ measurements given by the tree~$T$ on state~$Q$ and
having obtained the correct outcomes~$Y_{\Pi(1)}Y_{\Pi(2)}\dotsb Y_{\Pi(l-1)}%
$. \ These outcomes determine the index~$\Pi(l)$ of the next bit that may be
learned. \ By the Chain Rule, for any~$l\in\left\{  1,\ldots,k-1\right\}  $,
\begin{align*}
& \operatorname*{I}\!\left(  Y_{\Pi(l)}\dotsb Y_{\Pi(k)}:Q\;|\;Y_{\Pi(1)}%
Y_{\Pi(2)}\dotsb Y_{\Pi(l-1)}\right) \\
&  =\operatorname*{I}\!\left(  Y_{\Pi(l)}:Q\;|\;Y_{\Pi(1)}Y_{\Pi(2)}\dotsb
Y_{\Pi(l-1)}\right)  +\operatorname*{I}\!\left(  Y_{\Pi(l+1)}\dotsb Y_{\Pi
(k)}:Q\;|\;Y_{\Pi(1)}Y_{\Pi(2)}\dotsb Y_{\Pi(l)}\right)  .
\end{align*}
Let~$E$ be the operator associated with the vertex~$V:=Y_{\Pi(1)}Y_{\Pi
(2)}\dotsb Y_{\Pi(l-1)}$. \ By hypothesis, the measurement~$E$
predicts the bit~$Y_{\Pi(l)}$ with error at most~$p_{V}$. \ Using the Fano
Inequality, and averaging over the prefix~$V$, we get
\[
\operatorname{I}\!\left(  Y_{\Pi(l)}:Q\;|\;Y_{\Pi(1)}Y_{\Pi(2)}\dotsb
Y_{\Pi(l-1)}\right)  \geq\mathbb{E}_{V}(1-\operatorname*{H}(p_{V}))
\enspace.
\]
Applying this repeatedly for~$l\in\left\{  1,\ldots,k-1\right\}  $, we get
\begin{align*}
\operatorname{I}(Y:Q)  &  =\operatorname{I}\!\left(  Y_{\Pi(1)}:Q\right)
+\operatorname{I}\!\left(  Y_{\Pi(2)}:Q\;|\;Y_{\Pi(1)}\right)  +\operatorname{I}%
\!\left(  Y_{\Pi(3)}:Q|Y_{\Pi(1)}Y_{\Pi(2)}\right) \\
& \quad \mbox{} +\dotsb+\operatorname{I}\!\left(  Y_{\Pi(k)}:Q\;|\;Y_{\Pi(1)}Y_{\Pi(2)}\dotsb
Y_{\Pi(k-1)}\right) \\
&  \geq\sum_{l=1}^{k}\frac{1}{2^{l}}\sum_{v\in\left\{  0,1\right\}  ^{l}%
}(1-\operatorname*{H}(p_{v}))\\
&  \geq(1-\operatorname*{H}(p))k \enspace,
\end{align*}
by concavity of the binary entropy function, and the definition of~$p$.
\end{proof}

\end{document}